\documentclass[journal,twoside]{IEEEtran}

\usepackage[T1]{fontenc}
\usepackage[utf8]{inputenc}
\usepackage{amsmath, amssymb, amsthm}
\usepackage{graphicx}
\usepackage{cite}
\usepackage{booktabs}
\usepackage{bm}
\usepackage{xcolor}

\newtheorem{theorem}{Theorem}
\newtheorem{proposition}{Proposition}
\newtheorem{corollary}{Corollary}

\newcommand{\kn}{\kappa_N}
\newcommand{\dn}{\Delta_N}
\newcommand{\bA}{\bm{A}}
\newcommand{\bB}{\bm{B}}
\newcommand{\bR}{\bm{R}}
\newcommand{\bq}{\bm{q}}
\newcommand{\bp}{\bm{p}}
\newcommand{\bP}{\bm{P}}
\newcommand{\bPbar}{\bar{\bm{P}}}
\newcommand{\bSigma}{\bm{\Sigma}}
\newcommand{\ehat}{\hat{\bm{e}}}
\newcommand{\trace}{\operatorname{tr}}

\begin{document}

\title{$\kappa$: A Geometry-Quality Metric Complementary to GDoP for Closed-Form TDoA Multilateration}

\author{Abeer~Nasir~Chaudhry,
        Salman~Liaquat,~\IEEEmembership{Member,~IEEE},
        and Muhammad~Mohsin~Khadim%
\thanks{A.~N.~Chaudhry is with the American University of Sharjah, Sharjah, United Arab Emirates (e-mail: b00085496@aus.edu).}%
\thanks{S.~Liaquat is with the School of Electrical and Electronic Engineering, Universiti Sains Malaysia, Penang, Malaysia and National Aerospace Science and Technology Park, Islamabad, Pakistan (e-mail: salmanliaquat@ieee.org).}%
\thanks{M.~M.~Khadim is with the National University of Sciences and Technology, Islamabad, Pakistan (e-mail: mohsinkhadim59@gmail.com).}%
\thanks{Manuscript prepared \today. Corresponding author: Salman Liaquat.}%
}

\markboth{IEEE Transactions on Aerospace and Electronic Systems}{}

\maketitle

\begin{abstract}
The Geometric Dilution of Precision (GDoP) characterizes the noise sensitivity of a Time-Difference-of-Arrival (TDoA) localization system, but does not capture every way the analytical multilateration solution can become ill-conditioned. We introduce a complementary geometry-quality metric $\kappa$, the leading coefficient of the closed-form TDoA solver's quadratic, and derive its $N$-dimensional generalization through a vectorized formulation. Two closed-form algebraic identities relate $\kappa$ to the Jacobian determinant of the measurement model and to the quadratic's discriminant, establishing that the system exhibits exactly two distinct singularity loci: branch divergence and the Jacobian/branch-merge locus flagged by GDoP. A Cram\'{e}r--Rao-bound-linked closed form for the noise sensitivity $\sigma_\kappa$ under the standard Gaussian ToA model is validated against Monte~Carlo to 2\% median relative error. An empirical atlas over a dimensionless geometry parameter space confirms both identities at machine precision and shows that $\kappa$-bad regions and GDoP-bad regions are non-trivially disjoint in target space, establishing the two metrics as genuinely complementary. A case study on a four-node operational array, with per-sensor time of arrival (ToA) noise estimated empirically from Automatic Dependent Surveillance Broadcast (ADS-B)-paired over-the-air captures, shows that the theory-predicted threshold and a Monte-Carlo-measured operational threshold agree on the per-subsystem ordering at the deployment noise level. Their ratio is approximately constant across the three two-dimensional subsystems, serving as a deployment-specific calibration constant between the algebraic $\kappa$-noise floor and the downstream operational threshold, analogous in spirit to the standard relation linking GDoP to the circular error probable.
\end{abstract}


\section{Introduction}
\IEEEPARstart{T}{ime}-Difference-of-Arrival (TDoA) multilateration is a foundational technique in modern passive surveillance, supporting applications that span civil air-traffic management via Automatic Dependent Surveillance Broadcast (ADS-B) multilateration networks \cite{schaeferjonas2025, sonnleitner2025, osypiuksurma2025}, wide-area Global Navigation Satellite System (GNSS) jamming and spoofing detection \cite{gattis2026}, maritime moving-target tracking with multistatic passive radar \cite{nassosanti2024}, and electronic-warfare emitter geolocation. The architectural appeal of a TDoA system is parsimony: a network of time-synchronized receivers, no transmitter, and no specialized angle-of-arrival hardware. Its principal challenge is geometric: the localization accuracy, the uniqueness of the position solution, and the numerical stability of the underlying solver all depend on where the receivers sit relative to the target.

Two families of solvers dominate the TDoA literature. Iterative methods (e.g., Taylor-series least-squares \cite{foy1976}) linearize the nonlinear TDoA equations around an initial guess and refine through successive updates; they sidestep the algebraic degeneracies analyzed here at the cost of initial-guess dependence and possible divergence in near-collinear configurations. Closed-form analytical methods avoid iteration entirely: by squaring and combining the TDoA equations, eliminating cross-terms via a coordinate transformation, and reducing the system to a single scalar polynomial, they obtain candidate solutions in finite steps \cite{bancroft1985, smithabel1987, smithabel1987tassp, schauRobinson1987, friedlander1987, fang1990, chanho1994, huangBenestyElkoMersereau2001, sunhowan2019, hubacek2022, linwang2024, xingsun2024, inamdar2025}. The closed-form route is preferred in operational systems for its determinism, computational efficiency, and insensitivity to initial-guess pathology. It pays a price, however: the algebraic manipulations introduce degenerate configurations specific to the closed-form route that the iterative route does not encounter, because it never forms the offending polynomial in the first place.

The geometric quality of a TDoA receiver array has been studied extensively. The Cram\'{e}r--Rao bound (CRLB) framework \cite{torrieri1984} provides a lower bound on the position-estimation variance achievable by any unbiased estimator, and the Geometric Dilution of Precision (GDoP) is its dimensionless companion \cite{sunhowan2019, pinepine2021}. Recent work has extended these tools to moving-source TDoA--FDoA configurations \cite{yangzheng2024}, to sensor-placement design accounting for sensor-location uncertainty \cite{zhanghan2025}, to four-receiver three-dimensional TDoA with explicit treatment of the convergence-sphere parameter as a scalar geometry-quality indicator \cite{diezgonzalez2019}, to closed-form three-dimensional extensions with hybrid measurement types \cite{linwang2024, xingsun2024}, to the well-posedness of source-localization closed forms via semi-definite-programming and generalized-trust-region reductions \cite{beckstoicali2008}, and to fault-tolerant hyperbolic localization with sensor consensus \cite{simonzachar2024}. The unifying assumption across this literature is that the relevant geometric pathology is captured by the Jacobian of the measurement map: GDoP diverges precisely when this Jacobian becomes rank-deficient.

This assumption is sufficient if the only failure mode of a TDoA solver is noise amplification near a rank-deficient Jacobian. For closed-form algebraic solvers, it is not. The coordinate-transformation and squaring steps that linearize the TDoA equations introduce an algebraic-singularity structure that does not coincide with Jacobian rank-deficiency: the resulting polynomial in the unknown range can degenerate in a manner that has no counterpart in the Jacobian-based picture. Hub\'{a}{\v c}ek \emph{et al.}\ \cite{hubacek2022}, working with the two-dimensional three-sensor case, noted this degeneracy and denoted its indicator $M = B^2 + D^2 - 1$, observing that the analytical solver exhibits root-divergence on the locus $M = 0$. They further reported empirically that the localization error grows sharply in the same neighborhood, writing of these high-error regions that ``there are just the areas where the $M = 0$ condition is satisfied'' \cite[p.~231]{hubacek2022}. This is the closest the prior literature comes to using $M$ as a quality indicator: the empirical coincidence between the $M = 0$ locus and the high-error pattern is identified, but $M$ is treated as a solvability condition with a degenerate locus, not as a metric whose value is propagated under measurement noise. Algebraic-geometric treatments of the TDoA map have characterized the bifurcation locus where the analytical solver degenerates, together with the related Jacobian-discriminant locus, for the planar three-receiver case \cite{compagnoni2014inverse, compagnoni2014arxiv}, with partial extensions to the three-dimensional range-map setting \cite{compagnoni2017jns}. A subsequent statistical treatment of the same singular set under range-difference noise appears in \cite{compagnoni2016arxiv}, the closest prior framing to the noise-propagation approach we adopt for $\sigma_\kappa$ in Section~\ref{sec:theorem3}. These results, however, treat the degeneracy as a mathematical object to be classified or as a statistical model of the source, not as a quantity to be propagated and reported alongside GDoP at runtime. To the best of our literature search, no prior publication names and operationally packages the leading-coefficient indicator as a general-$N$ geometry-quality metric in solver-state coordinates, derives a closed-form noise sensitivity for it, or quantifies its empirical complementarity with GDoP across the dimensionless geometry space. The need for this characterization is operational: a four-node array can simultaneously support a three-dimensional solver and several three-sensor two-dimensional subsystem solvers, each carrying its own closed-form polynomial whose algebraic-degeneracy locus must be flagged independently of the GDoP layer.

This paper closes that gap. Treating the leading-coefficient indicator as a first-class geometry-quality metric in its own right, we generalize it to $N$-dimensional configurations through a vectorized derivation, link it to the Jacobian determinant via two new closed-form algebraic identities, and demonstrate, both theoretically and through an extensive empirical atlas, that it is genuinely complementary to GDoP. We denote the generalized coefficient $\kn$ and refer to its associated diagnostic as the \emph{$\kappa$ layer} of multilateration geometry quality, distinct from the noise-sensitivity GDoP layer.\footnote{Poisel \cite[\S5.2.7]{poisel2012} uses the symbol $\bar{\kappa}$ for the TDoA-Jacobian matrix that defines GDoP, $g = \sqrt{\trace((\bar{\kappa}^{\top}\bar{\kappa})^{-1})}$. Our $\kn$ is unrelated: a scalar leading coefficient of the $K$-quadratic, not a Jacobian matrix.} Used jointly, the two metrics yield a $2 \times 2$ classification of target positions that refines the GDoP-only picture and resolves operational ambiguities about per-subsystem solver health in multi-receiver arrays.

This paper makes three main contributions:
\begin{enumerate}
\item A unified $N$-dimensional vectorized derivation of $\kn = \|\bB\|^2 - 1$ as the leading coefficient of the analytical TDoA $K$-quadratic for $N{+}1$ sensors in $\mathbb{R}^N$ (the 2D case recovers Hub\'{a}{\v c}ek's $M$; the 3D case the four-node operational solver), together with two closed-form algebraic identities (T2.1, T2.2) factoring $\det \bm{J}'$ and the $K$-quadratic discriminant $\dn$ in the solver-state quantities $(\kn, K, \|\bA\|^2)$ for general~$N$. The singularity loci themselves were identified in the planar three-receiver case by Compagnoni and Notari \cite{compagnoni2014inverse, compagnoni2014arxiv}; our contribution is the closed-form factorization in solver-state coordinates that is directly computable from TDoA data and extends to general~$N$.
\item A CRLB-linked closed-form expression for $\sigma_\kappa(\bq, \sigma_t, \text{geometry})$ from a chain-rule derivation under the standard Gaussian ToA noise model, validated against Monte~Carlo to within $2\%$ median relative error.
\item A dimensionless empirical atlas in $(\beta, \gamma, \tilde{r}, \theta)$ confirming the identities at machine precision and demonstrating that $\kn$-bad and GDoP-bad regions are non-trivially disjoint in target space, together with a case-study application on a deployed four-node array showing that the algebraic threshold of Corollary~\ref{cor:threshold}, at the deployment's empirically estimated $\sigma_t$, predicts the correct per-subsystem operational ordering with a near-constant calibration ratio between the algebraic and operational thresholds.
\end{enumerate}

The remainder of the paper is organized as follows. Section~\ref{sec:background} reviews the TDoA measurement model, the closed-form solver structure, and GDoP. Section~\ref{sec:theory} develops the theoretical framework. Section~\ref{sec:atlas} presents the empirical atlas. Section~\ref{sec:coupling} quantifies the empirical relationship between $\kn$ and GDoP. Section~\ref{sec:case-study} applies the framework to an operational four-node array. Section~\ref{sec:discussion} discusses practical implications, and Section~\ref{sec:conclusion} concludes.

\section{Background and System Model}
\label{sec:background}

Throughout the paper, vectors are typeset in boldface lowercase ($\bp_i$, $\bq$), matrices in boldface uppercase ($\bm{P}$, $\bPbar$, $\bSigma_T$), and unit vectors with a hat ($\ehat_i$). Subscript $0$ denotes the central reference sensor (CRS) and subscripts $1, \ldots, N$ denote the secondary reference sensors (SRSs). Primed quantities ($\bq'$, $\bp_i'$, $\bm{J}'$) live in the canonical frame.

\subsection{TDoA measurement model and noise}
\label{sec:model}

Consider a network of $N+1$ time-synchronized sensors at positions $\bp_0, \bp_1, \ldots, \bp_N \in \mathbb{R}^N$, with $\bp_0$ designated CRS. A target at unknown position $\bq \in \mathbb{R}^N$ radiates a signal whose times of arrival (ToAs) $t_0, t_1, \ldots, t_N$ at the $N+1$ sensors are recorded. The pairwise TDoA range-differences are
\begin{equation}
R_i \;=\; c\,(t_i - t_0) \;=\; \|\bq - \bp_i\| - \|\bq - \bp_0\|, \quad i = 1, \ldots, N,
\label{eq:Ri}
\end{equation}
where $c$ is the signal-propagation speed. Without loss of generality we adopt the world frame in which the CRS lies at the origin, $\bp_0 = \bm{0}$; this rigid translation leaves \eqref{eq:Ri} and the Jacobian of $\bR$ with respect to $\bq$ unchanged, and reduces $K := \|\bq - \bp_0\|$ to $\|\bq\|$ throughout. Each $R_i$ constrains the target to one sheet of a hyperboloid of revolution (a hyperbola, in 2D) with foci at $\bp_0$ and $\bp_i$. With $N+1$ sensors in $\mathbb{R}^N$, the intersection of $N$ such hyperboloids generically contains exactly two candidate target positions \cite{fang1990, hubacek2022, inamdar2025}, only one of which is the true target; distinguishing the true point from the phantom is the object of any subsequent disambiguation step.

We adopt the standard per-sensor noise model: each ToA $t_i$ is corrupted by zero-mean Gaussian noise of variance $\sigma_t^2$, independent across sensors. The per-sensor variance $\sigma_t$ subsumes timestamping jitter, clock-synchronization error, and signal-processing-induced ToA uncertainty. The induced covariance of the range-difference vector $\bR = (R_1, \ldots, R_N)^{\top}$ follows from $\operatorname{cov}(t_i - t_0, t_j - t_0) = \sigma_t^2 (\delta_{ij} + 1)$ and is
\begin{equation}
\operatorname{cov}(\bR) \;=\; c^2 \sigma_t^2 \, \bSigma_T, \qquad
\bSigma_T \;=\; \bm{I}_N + \bm{1}_N \bm{1}_N^{\top},
\label{eq:SigmaT}
\end{equation}
i.e., $\bSigma_T$ has 2 on the diagonal and 1 on the off-diagonals. In operational settings $\sigma_t$ is typically estimated empirically from per-sample-detrended TDoA residuals against ground-truth target positions, as $\sigma_t = \operatorname{std}(\text{residuals}) / \sqrt{2}$, where the $\sqrt{2}$ accounts for the differencing.

\subsection{Closed-form linearization and the canonical transformation}
\label{sec:linearization}

Equation \eqref{eq:Ri} is nonlinear in $\bq$ through the Euclidean norms. Closed-form algebraic methods \cite{bancroft1985, smithabel1987, fang1990, chanho1994, sunhowan2019, hubacek2022, linwang2024, xingsun2024, inamdar2025} avoid iteration by exploiting algebraic structure. Squaring \eqref{eq:Ri}: writing $r_i := \|\bq - \bp_i\|$ and $K := r_0 = \|\bq\|$, the identity $r_i^2 = (R_i + K)^2$ combined with $r_i^2 = \|\bq\|^2 - 2\,\bq \cdot \bp_i + \|\bp_i\|^2$ produces, after cancellation of the $\|\bq\|^2$ term that does not depend on $i$,
\begin{equation}
\bq \cdot \bp_i \;=\; \tfrac{1}{2}\bigl(\|\bp_i\|^2 - R_i^2\bigr) \;-\; R_i K, \quad i = 1, \ldots, N.
\label{eq:innerprod}
\end{equation}
Each equation in \eqref{eq:innerprod} is now linear in $\bq$ once $K$ is treated as an additional scalar unknown. Stacking the $N$ equations gives
\begin{equation}
\bm{P} \bq \;=\; \bm{\alpha} \;-\; K \bR,
\label{eq:linsys}
\end{equation}
where $\bm{P}$ has rows $\bp_i^{\top}$ and $\alpha_i := \tfrac{1}{2}(\|\bp_i\|^2 - R_i^2)$. The system has $N+1$ scalar unknowns in $N$ linear equations, leaving one residual degree of freedom that the geometric norm constraint $K = \|\bq\|$ will fix.

We apply an orthogonal transformation $\bar{\bm{R}}$, constructed via successive Givens rotations, that maps the world frame to a canonical frame in which $\bp_0' = \bm{0}$ and the SRS positions are lower triangular ($\bp_i'$ has only its first $i$ coordinates nonzero). The linear system \eqref{eq:linsys} becomes
\begin{equation}
\bPbar \bq' \;=\; \bm{\alpha}' \;-\; K \bR,
\label{eq:linsys-canon}
\end{equation}
with $\bPbar$ the $N \times N$ lower-triangular canonical sensor matrix, $\alpha_i' = \tfrac{1}{2}(\|\bp_i'\|^2 - R_i^2)$, and $\bR$ invariant under the (Euclidean) rotation. $\bPbar$ is invertible by forward substitution whenever the sensors are in general position (i.e., not all in a common $(N-1)$-dimensional subspace). The closed-form solution to \eqref{eq:linsys-canon} is the affine position parameterization
\begin{equation}
\bq' \;=\; \bA + K \cdot \bB,
\label{eq:param}
\end{equation}
with $\bA = \bPbar^{-1} \bm{\alpha}'$ and $\bB = -\bPbar^{-1} \bR$. The transformation makes \eqref{eq:linsys} numerically well-posed and exposes the affine $K$-dependence of \eqref{eq:param} ready for the final geometric constraint. The price of this route is that the squaring step in \eqref{eq:innerprod} introduces algebraic degeneracies distinct from those of the iterative route; characterizing them is the subject of Section~\ref{sec:theory}.

\subsection{The $K$-quadratic}
\label{sec:Kquad-derivation}

To close the system into a single scalar equation, we impose $K = \|\bq'\|$, i.e., $K^2 = \|\bq'\|^2$. Substituting \eqref{eq:param},
\begin{equation*}
K^2 \;=\; \|\bA + K\bB\|^2 \;=\; \|\bA\|^2 + 2 K (\bA \cdot \bB) + K^2 \|\bB\|^2,
\end{equation*}
which rearranges to the \emph{$K$-quadratic}:
\begin{equation}
\kn K^2 \;+\; 2(\bA \cdot \bB) K \;+\; \|\bA\|^2 \;=\; 0, \qquad \kn \;:=\; \|\bB\|^2 - 1.
\label{eq:Kquad}
\end{equation}
The two roots of \eqref{eq:Kquad} correspond to the two intersection points of the $N$ hyperboloids; pairing each root with \eqref{eq:param} and inverting the canonical transformation yields the two candidate target positions in world coordinates, with a geometric disambiguation rule (typically $K > 0$, plus a sequential-measurement or prior-position consistency check) selecting the true target.

The leading coefficient $\kn$, the central object of this paper, depends on the TDoA measurements and canonical sensor geometry but not on $K$ itself. The 2D three-sensor case has $\bB = (B, D)^{\top}$ and $\kn = B^2 + D^2 - 1$, which is the indicator $M$ studied by Hub\'{a}{\v c}ek \emph{et al.}~\cite{hubacek2022}. The 3D four-sensor case has $\bB = (B_1, B_2, B_3)^{\top}$ and $\kn = B_1^2 + B_2^2 + B_3^2 - 1$, the leading coefficient of the closed-form solver employed in the operational four-node array of Section~\ref{sec:case-study}. The vectorized definition $\kn = \|\bB\|^2 - 1$ is the natural $N$-dimensional generalization.

\subsection{Geometric Dilution of Precision}
\label{sec:gdop}

Let $\bm{J}$ denote the Jacobian of the range-difference measurement vector $\bR$ with respect to target position $\bq$. Differentiating \eqref{eq:Ri} gives row $i$ as the line-of-sight unit-vector difference $\bm{J}_i = \ehat_i - \ehat_0$, where $\ehat_i = (\bq - \bp_i)/\|\bq - \bp_i\|$. Under the noise model of \eqref{eq:SigmaT}, the CRLB on the variance of any unbiased position estimator gives
\begin{align}
\sigma_{\text{pos}}(\bq) &\;=\; c \, \sigma_t \cdot \tilde{G}(\bq), \nonumber\\
\tilde{G}(\bq) &\;=\; \sqrt{\operatorname{tr}\!\left(\left(\bm{J}^{\top} \bSigma_T^{-1} \bm{J}\right)^{-1}\right)},
\label{eq:GDoP}
\end{align}
with $\tilde{G}$ the dimensionless GDoP. The inner quantity $(\bm{J}^{\top} \bSigma_T^{-1} \bm{J})^{-1}$ is the Fisher-information-matrix inverse for an unbiased estimator of $\bq$ under the noise model \eqref{eq:SigmaT}; the dimensionless $\tilde{G}$ is therefore the trace-CRLB baseline (in units of metres per $c\sigma_t$). $\tilde{G}$ diverges precisely when $\bm{J}$ becomes rank-deficient, i.e., when the rows $\ehat_i - \ehat_0$ become linearly dependent \cite{torrieri1984, sunhowan2019, pinepine2021, yangzheng2024, zhanghan2025}. The dimensionless $\tilde{G}$ has a direct operational counterpart in the circular error probable at the $50$\% confidence level \cite[eq.~(74)]{torrieri1984}, \cite[\S 9.2]{odonoughue2020},
\begin{equation}
\mathrm{CEP}_{50}(\bq) \;\approx\; 0.75 \, c \, \sigma_t \, \tilde{G}(\bq),
\label{eq:CEP50}
\end{equation}
valid in the 2D approximately-isotropic-error regime. The two are proportional at fixed $\sigma_t$; we report results in dimensionless $\tilde{G}$ for geometry-only analysis and in $\mathrm{CEP}_{50}$ wherever an operational threshold is implied.

\section{Theoretical Framework}
\label{sec:theory}

\subsection{Theorem 1: $K$-quadratic identity}
\label{sec:theorem1}

\begin{theorem}[Algebraic identity for $\kn$]
Under the setup of Sections~\ref{sec:model}--\ref{sec:linearization}, at any physical target the analytical TDoA solution satisfies the position parameterization $\bq' = \bA + K \cdot \bB$ of \eqref{eq:param} together with the $K$-quadratic identity \eqref{eq:Kquad}, where $\kn = \|\bB\|^2 - 1$.
\label{thm:1}
\end{theorem}
\begin{IEEEproof}
Squaring \eqref{eq:Ri} and collecting terms in $\bq$ yields the linear system $\bP \bq = \bm{\alpha} - K \bR$ in the unrotated frame. After the canonical transformation $\bar{\bm{R}}$, the system becomes $\bPbar \bq' = \bm{\alpha}' - K \bR$ with $\bPbar$ lower triangular. Solving for $\bq'$ gives \eqref{eq:param}. Imposing $\|\bq'\|^2 = K^2$ produces
\[
\|\bA\|^2 + 2 K (\bA \cdot \bB) + K^2 \|\bB\|^2 = K^2,
\]
which is \eqref{eq:Kquad}.
\end{IEEEproof}

\paragraph{Runtime computation} Theorem~\ref{thm:1} prescribes a direct recipe for $\kn$ without prior knowledge of the source position. Given $\bp_0, \ldots, \bp_N$ and $\{\tau_i\}_{i=1}^N$: (i) compute $R_i = c\tau_i$ and $\bPbar$ by Givens rotations of $\{\bp_i - \bp_0\}$; (ii) solve $\bPbar \bA = \bm{\alpha}'$ and $\bPbar \bB = -\bR$ by forward substitution; (iii) read $\kn = \|\bB\|^2 - 1$. Step~(i) is once per array; steps~(ii)--(iii) per sample cost $O(N^2)$. Reporting $\kn$ alongside the solution adds no work beyond a single dot product on $\bB$.

\subsection{Algebraic identities (T2.1, T2.2): $\kn \leftrightarrow \bm{J}$ and $\kn \leftrightarrow \dn$}
\label{sec:theorem2}

Proposition~\ref{thm:T21} (T2.1) factors the determinant of the TDoA Jacobian into a sensor-only piece and a single target-dependent piece that vanishes precisely when $\kn K^2 = \|\bA\|^2$; Theorem~\ref{thm:T22} (T2.2) shows that this same target-dependent factor governs the discriminant of the $K$-quadratic, so the ``two roots coincide'' picture (branch merge) and the ``inverse map is non-invertible'' picture (Jacobian singularity) describe the same algebraic condition on the target. T2.1 is presented as a Proposition because, given the augmented hyperbolic linear system of \cite{schauRobinson1987, fang1990, smithabel1987tassp, chanho1994}, the factorization is a matrix-determinant-lemma manipulation; the substantive contribution is the explicit form in solver-state coordinates $(\kn, K, \|\bA\|^2)$, which is the form that connects directly to the discriminant identity T2.2 and to the closed-form $\sigma_\kappa$ of Section~\ref{sec:theorem3}.

\begin{proposition}[Jacobian determinant identity (T2.1)]
In the canonical frame, the $N \times N$ Jacobian $\bm{J}'$ with rows $\bm{J}'_i = \ehat_i' - \ehat_0'$ has determinant
\begin{equation}
\det \bm{J}'(\bq') = (-1)^{N+1} \frac{\det \bPbar}{2 K^2 \prod_{i=1}^N r_i} \, (\kn K^2 - \|\bA\|^2),
\label{eq:T21}
\end{equation}
where $r_i = \|\bq - \bp_i\|$ and the product runs over the $N$ SRSs.
\label{thm:T21}
\end{proposition}
\begin{IEEEproof}
Row $i$ of $\bm{J}'$ in canonical frame is $\bq'^{\top}(1/r_i - 1/K) - \bp_i'^{\top}/r_i$. Stacking gives $\bm{J}' = \bm{u} \bq'^{\top} - \bm{D} \bPbar$, where $\bm{u}_i = 1/r_i - 1/K$ and $\bm{D} = \operatorname{diag}(1/r_i)$. Using $r_i = K + R_i$ to identify $\bm{D}^{-1} \bm{u} = -\bR/K$, and applying the matrix determinant lemma to $\bPbar + (1/K)\bR\bq'^{\top}$ yields
\[
\det \bm{J}' = (-1)^N \frac{\det \bPbar}{\prod r_i} \left(1 - \frac{\bq'^{\top} \bB}{K}\right).
\]
Substituting $\bq'^{\top} \bB = \bA \cdot \bB + K\|\bB\|^2$ from \eqref{eq:param} and using $\kn = \|\bB\|^2 - 1$ gives
\[
1 - \frac{\bq'^{\top}\bB}{K} = -\kn - \frac{\bA \cdot \bB}{K}.
\]
The $K$-quadratic \eqref{eq:Kquad} gives $\bA \cdot \bB = -(\kn K^2 + \|\bA\|^2)/(2K)$, hence
\[
-\kn - \frac{\bA \cdot \bB}{K} = -\frac{\kn K^2 - \|\bA\|^2}{2K^2}.
\]
Combining the leading $(-1)^N$ factor with this minus sign yields \eqref{eq:T21}.
\end{IEEEproof}

\begin{theorem}[Discriminant identity (T2.2)]
Assume the sensors are in general position ($\det \bPbar \neq 0$) and the target is not coincident with any sensor (all $r_i > 0$). On the constraint $K = \|\bq'\|$, the $K$-quadratic discriminant satisfies
\begin{equation}
\dn := (\bA \cdot \bB)^2 - \kn \|\bA\|^2 = \frac{(\kn K^2 - \|\bA\|^2)^2}{4 K^2}.
\label{eq:T22}
\end{equation}
Under these assumptions, $\dn = 0$ if and only if $\det \bm{J}' = 0$ at the same target position.
\label{thm:T22}
\end{theorem}
\begin{IEEEproof}
From \eqref{eq:Kquad}, $\bA \cdot \bB = -(\kn K^2 + \|\bA\|^2)/(2K)$, so $(\bA \cdot \bB)^2 = (\kn K^2 + \|\bA\|^2)^2/(4 K^2)$. Expanding:
\begin{align*}
\dn &= \frac{(\kn K^2 + \|\bA\|^2)^2}{4 K^2} - \kn \|\bA\|^2 \\
    &= \frac{(\kn K^2)^2 + 2\kn K^2 \|\bA\|^2 + \|\bA\|^4 - 4 \kn K^2 \|\bA\|^2}{4 K^2} \\
    &= \frac{(\kn K^2 - \|\bA\|^2)^2}{4 K^2},
\end{align*}
which is \eqref{eq:T22}. Under the stated regularity assumptions, $\det \bPbar / (2K^2 \prod r_i) \neq 0$, so by \eqref{eq:T21} $\det \bm{J}' = 0$ if and only if $\kn K^2 - \|\bA\|^2 = 0$; by \eqref{eq:T22} the same condition is equivalent to $\dn = 0$.
\end{IEEEproof}

\subsection{Singularity structure}
\label{sec:singularity-structure}

Combining Proposition~\ref{thm:T21} and Theorem~\ref{thm:T22} yields a clean classification: the multilateration system has exactly \emph{two} distinct (codimension-one) singularity loci in target space.

\begin{enumerate}
\item \textbf{$\kn(\bq) = 0$:} the leading coefficient of the $K$-quadratic vanishes; the quadratic degenerates to a linear equation and one of its two roots escapes to infinity. We refer to this as the \emph{branch-divergence} locus.
\item \textbf{$\det \bm{J}(\bq) = 0$:} equivalently (by Theorem \ref{thm:T22}) $\dn(\bq) = 0$; the two roots of the $K$-quadratic coincide (\emph{branch-merge}), and the linearized inverse map from TDoA to position is rank-deficient (\emph{Jacobian} interpretation). GDoP diverges on this locus.
\end{enumerate}

The two loci are distinct: at $\kn = 0$ (off the intersection set) $\det \bm{J}$ is generically nonzero, and conversely. By \eqref{eq:T21}, $\det \bm{J} = 0$ requires $\kn K^2 = \|\bA\|^2$; intersecting this with $\kn = 0$ forces $\|\bA\|^2 = 0$, so the two loci meet on the \emph{doubly-singular set} characterized by $\kn = 0$ and $\|\bA\|^2 = 0$. This set has codimension two and corresponds to specific collinearity configurations of target and sensor pair. The Jacobian-singular locus along sensor-pair baselines is a known phenomenon (Torrieri \cite[Sec.~V, eq.~(103)]{torrieri1984}; the resulting lobed CRLB-driven CEP pattern for a three-sensor array appears in textbook form in O'Donoughue \cite[Fig.~11.6]{odonoughue2020}).

\subsection{Theorem 3: CRLB-linked noise sensitivity of $\kn$}
\label{sec:theorem3}

Theorem~\ref{thm:T3} treats $\kn$ as a random variable induced by Gaussian per-sensor ToA noise and gives its first-order variance in closed form, propagated through the same noise model that underlies the CRLB.

\begin{theorem}[Closed-form $\sigma_\kappa$]
Under the per-sensor independent-noise convention $\operatorname{cov}(t_i) = \sigma_t^2$ giving $\operatorname{cov}(\bR) = c^2 \sigma_t^2 \bSigma_T$ as in \eqref{eq:SigmaT}, the linearized variance of $\kn$ at target $\bq$ is
\begin{equation}
\sigma_\kappa^2(\bq) = (\nabla_{\bR} \kn)^{\top} \, c^2 \sigma_t^2 \bSigma_T \, (\nabla_{\bR} \kn),
\label{eq:T3}
\end{equation}
with $\nabla_{\bR} \kn = -2 \bPbar^{-\top} \bB$.
\label{thm:T3}
\end{theorem}
\begin{IEEEproof}
Since $\kn = \|\bB(\bR)\|^2 - 1 = \bR^{\top} (\bPbar^{-\top} \bPbar^{-1}) \bR - 1$, the gradient is $\nabla_{\bR} \kn = 2 (\bPbar^{-\top} \bPbar^{-1}) \bR$. Substituting $\bPbar^{-1} \bR = -\bB$ gives $\nabla_{\bR} \kn = -2 \bPbar^{-\top} \bB$. The linearized (first-order Taylor) variance of any smooth scalar $f(\bR)$ of a Gaussian $\bR \sim \mathcal{N}(\bar{\bR}, \bm{\Sigma})$ at the mean is $\operatorname{Var}(f) \approx (\nabla_{\bR} f)^{\top} \bm{\Sigma} (\nabla_{\bR} f)$, the standard delta-method variance propagation used throughout the CRLB-linked TDoA literature \cite{chanho1994, sunhowan2019, beckstoicali2008}. Setting $f = \kn$ and $\bm{\Sigma} = c^2 \sigma_t^2 \bSigma_T$ gives \eqref{eq:T3}; novelty lies in the closed-form gradient $\nabla_{\bR} \kn = -2 \bPbar^{-\top} \bB$, not in the variance-propagation machinery itself.
\end{IEEEproof}

\begin{corollary}[Threshold from theory]
A useful per-subsystem threshold for flagging samples on which $\kn$ is statistically indistinguishable from zero takes the form $\varepsilon_s \geq k \cdot \sigma_\kappa^{(\text{typ})}$, where $\sigma_\kappa^{(\text{typ})}$ is the median of $\sigma_\kappa$ over the operational target envelope for subsystem $s$. The choice $k = 3$ corresponds to the standard three-sigma criterion: a measured $|\kn| < 3 \sigma_\kappa^{(\text{typ})}$ has probability less than $0.3\%$ of arising from a non-degenerate $\kn(\bq) \gg 0$ target, so samples falling below $\varepsilon_s$ are flagged as algebraically ill-conditioned. The factor $k$ may be tuned higher for more conservative gating or lower at the cost of higher false-flag rate.
\label{cor:threshold}
\end{corollary}

\subsection{Bridge to GDoP}
\label{sec:bridge-to-gdop}

By construction, GDoP diverges as $\det \bm{J} \to 0$; by Theorem \ref{thm:T22}, this locus coincides with $\dn = 0$. The two independent geometry-quality diagnostics are therefore $\kn$ and $\det \bm{J}$ (equivalently $\dn$ or GDoP), giving the $2 \times 2$ classification:
\begin{center}
\begin{tabular}{l|cc}
& GDoP OK & GDoP bad \\ \hline
$\kn$ OK & well-conditioned & branch-merge regime \\
$\kn$ bad & branch-divergence regime & doubly-singular set \\
\end{tabular}
\end{center}
Targets in the off-diagonal regimes are detected by one diagnostic but not the other; these are the cases where the $\kn$ layer adds information beyond GDoP.

\section{Empirical Atlas in Dimensionless Coordinates}
\label{sec:atlas}

To make the analysis site-agnostic, we work in dimensionless coordinates. After the canonical transformation, the 2D geometry is fully specified by canonical coefficients $(a, b, c)$ with SRSs at $(a, 0)$ and $(b, c)$; we normalize by $a$ to obtain
\begin{equation}
\beta = b/a, \quad \gamma = c/a,
\end{equation}
and likewise the dimensionless target polar coordinates $\tilde{r} = r/a$ and $\theta$. The atlas is computed on a $20 \times 20 \times 60 \times 60 = 1.44 \times 10^6$ grid with $\beta, \gamma \in [-1.5, 1.5]$ (20 points linear each), $\tilde{r} \in [0.2, 10]$ (60 points log-spaced), and $\theta \in [0, 2\pi)$ (60 points linear). At each grid point we record $\kn$, $\dn$, $\det \bm{J}$, $\tilde{G}$, and the theoretical $\sigma_\kappa$ at $\sigma_t = 1$.

\subsection{Numerical verification of Theorems 1--3}
\label{sec:numerical-verif}

Identity \eqref{eq:T21} was checked at every atlas point by comparing $\det \bm{J}'$ computed directly from $\ehat_i' - \ehat_0'$ against the right-hand-side of \eqref{eq:T21}. The maximum relative residual across the $1.44 \times 10^6$ atlas points is $2.2 \times 10^{-13}$, with median residual $3.5 \times 10^{-17}$, both at the floating-point round-off floor. An independent set of 1500 random configurations (1000 two-dimensional, 500 three-dimensional) gives a maximum relative residual of $6.8 \times 10^{-11}$, again at the conditioning-limited floor. Identity \eqref{eq:T21} is therefore exact and established empirically across six decades.

Identity \eqref{eq:T22} is verified in Figure~\ref{fig:T22}, which plots $|\det \bm{J}|$ against $\sqrt{\dn}$ for 5000 atlas points. The points fall on a single unit-slope line in log-log coordinates, with vertical scatter confined to the round-off floor.

\begin{figure}[t]
\centering
\includegraphics[width=0.8\columnwidth]{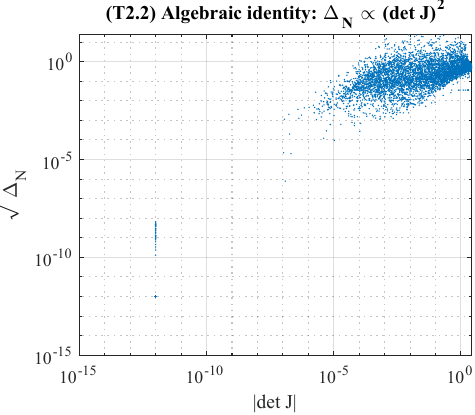}
\caption{Numerical verification of identity (T2.2): $|\det \bm{J}|$ vs $\sqrt{\dn}$ at $5000$ atlas points. Unit slope confirms that $\dn$ and $\det \bm{J}$ share the same zero locus.}
\label{fig:T22}
\end{figure}

For Theorem~\ref{thm:T3}, 200 random $(\beta, \gamma, \tilde{r}, \theta)$ configurations were drawn with the deployment-realistic baseline $a = 16.5$~km. At each configuration, the ToA vector was Monte-Carlo perturbed with per-sensor Gaussian noise of $\sigma_t = 30$~ns over $n_{\text{mc}} = 500$ realizations. Figure~\ref{fig:sigma_kappa_verify} compares closed-form prediction \eqref{eq:T3} against empirical standard deviation across roughly six decades. The points cluster tightly along $y = x$: median relative error $2.0\%$, $p_{95} = 5.8\%$, maximum $10.2\%$, consistent with the finite-sample Monte-Carlo standard error $1/\sqrt{2(n_{\text{mc}}-1)} \approx 3.2\%$. The linearized closed form \eqref{eq:T3} therefore agrees with the full nonlinear noise propagation to within sampling noise.

\begin{figure}[t]
\centering
\includegraphics[width=0.8\columnwidth]{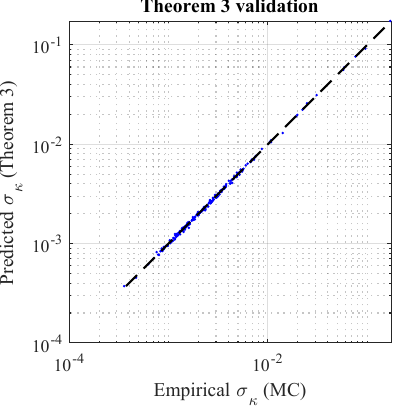}
\caption{Closed-form $\sigma_\kappa$ (Theorem~\ref{thm:T3}) vs Monte-Carlo empirical $\sigma_\kappa$ at $200$ random $(\beta, \gamma, \tilde{r}, \theta)$ configurations ($\sigma_t = 30$~ns, $n_{\text{mc}} = 500$). Dashed: $y = x$.}
\label{fig:sigma_kappa_verify}
\end{figure}

\begin{figure}[t]
\centering
\includegraphics[width=\columnwidth]{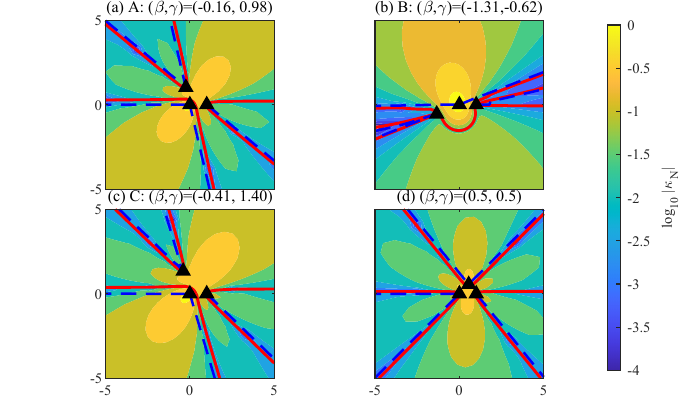}
\\[2pt] \footnotesize (top)~$\log_{10}|\kn|$ \\[6pt]
\includegraphics[width=\columnwidth]{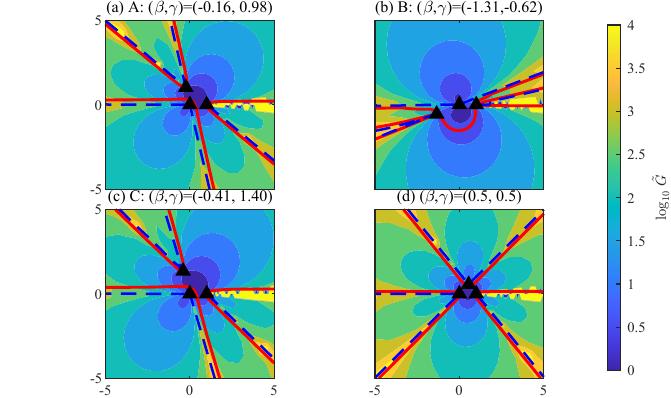}
\\[2pt] \footnotesize (bottom)~$\log_{10}\tilde{G}$
\caption{Stacked comparison of the two geometry-quality layers in the dimensionless target plane $(\tilde{r}\cos\theta, \tilde{r}\sin\theta)$ for four representative $(\beta, \gamma)$ points. \emph{Top:} $\log_{10}|\kn|$; \emph{bottom:} $\log_{10}\tilde{G}$ on the same four geometries. Inner panels (a)--(c) within each block: deployment subsystems A, B, C (Section~\ref{sec:case-study}); inner panel (d): generic contrast point. Singularity loci overlaid on both blocks: $\kn = 0$ (red solid), $\det \bm{J} = 0$ (blue dashed). Black triangles: sensors.}
\label{fig:iso_panels}
\end{figure}
\subsection{Singularity loci across the atlas}

Figure~\ref{fig:iso_panels} shows the two layers stacked: \emph{top:} $\log_{10}|\kn|$, and \emph{bottom:} $\log_{10}\tilde{G}$ over the dimensionless target plane at four representative geometries (the first three correspond to the deployment subsystems of Section~\ref{sec:case-study}; the fourth is a generic far-from-deployment point shown for contrast). The two singularity loci $\kn = 0$ (red solid) and $\det \bm{J} = 0$ (blue dashed) are overlaid on both blocks. The two curves are distinct (Proposition~\ref{thm:T21}, Theorem~\ref{thm:T22}) and meet only at the doubly-singular set on the sensor-pair baselines. Their enclosed regions differ visibly: target regions inside the red curve but outside the blue (and vice versa) are flagged by exactly one gate. All four geometries exhibit multi-lobed patterns of low-$|\kn|$ filaments emanating from the doubly-singular configurations, with lobe shape varying systematically with geometry; the $(\beta, \gamma) = (-1.31, -0.62)$ panel (subsystem B-like) shows the widest lobes.

The visual contrast between the two blocks of Fig.~\ref{fig:iso_panels} is the central empirical claim of the paper: in the bottom block, GDoP-bad regions track the Jacobian-singular locus (as expected, since $\tilde{G}$ diverges at $\det \bm{J} = 0$) but do \emph{not} track the $\kn$-bad lobes visible in the top block. In two of the four geometries, the GDoP-bad region and the $\kn$-bad lobes occupy nearly complementary sectors; in the other two, they partially overlap but each retains substantial regions not covered by the other.

\section{Empirical Relationship between $\kn$ and GDoP}
\label{sec:coupling}

We ask whether the $\kappa$ layer and the GDoP layer are statistically redundant (knowing one nearly determines the other) or \emph{complementary}, from two angles: the pooled joint distribution of $(\log_{10}|\kn|, \log_{10}\tilde{G})$ across the full atlas (Fig.~\ref{fig:correlation_heatmap}a), and the per-geometry Pearson correlation $\rho_{\kappa, G}$ aggregated over the 400 geometries (Fig.~\ref{fig:correlation_heatmap}b).

Panel~(a) shows the joint cloud is broad along both axes: at any fixed value of $\log_{10}\tilde{G}$, the corresponding $\log_{10}|\kn|$ spans several decades, and vice versa. Knowing one metric does not pin down the other. Panel~(b) shows $\rho_{\kappa, G}$ concentrated in $[-0.78, -0.52]$ with median $-0.65$ and interquartile range $[-0.69, -0.61]$. Three observations follow: (i) $\rho_{\kappa, G} < 0$ at every $(\beta, \gamma)$ tested; (ii) no geometry achieves $|\rho_{\kappa, G}| > 0.9$, so the two metrics are never tight enough for substitution; (iii) no geometry achieves $|\rho_{\kappa, G}| < 0.5$ either, so the two layers always share some signal (the doubly-singular collinearity events couple them). $\kn$ and GDoP are therefore \emph{neither redundant nor independent}: the empirical-statistics counterpart to the algebraic-independence statement of Section~\ref{sec:singularity-structure}.

\begin{figure*}[t]
\centering
\includegraphics[width=0.9\textwidth]{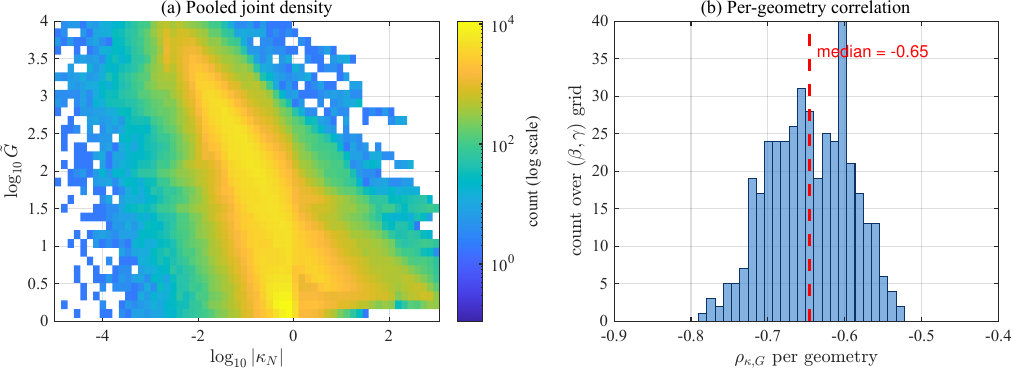}
\caption{(a) Joint log-count density of $(\log_{10}|\kn|, \log_{10}\tilde{G})$ across $1.44 \times 10^6$ atlas points. (b) Distribution of per-geometry Pearson $\rho_{\kappa, G}$ over the $20 \times 20$ $(\beta, \gamma)$ grid; range $[-0.78, -0.52]$, median $-0.65$ (red dashed).}
\label{fig:correlation_heatmap}
\end{figure*}

\subsection{Range and bearing dependence}

Figure~\ref{fig:radial_sweeps} traces $|\kn|$, $|\det \bm{J}|$, and $\tilde{G}$ along radial sweeps at three configurations representing the deployment subsystems. Panel C exhibits the diagnostic feature: at $\tilde{r} \approx 2$, $|\kn|$ drops by nearly five decades into a narrow algebraic-singularity dip while $\tilde{G}$ continues to vary gradually, and $|\det \bm{J}|$ shows the corresponding sharp dip predicted by Theorem~\ref{thm:T22}. This is a concrete instance of the off-diagonal regime: a target classified ``GDoP-acceptable'' at $\tilde{r} = 2$ on this bearing would be flagged ``$\kn$-bad'' by the algebraic layer, and a solver relying on GDoP alone would not see the warning. Panels A and B exhibit the typical pattern in which $\tilde{G}$ rises with range while $|\kn|$ and $|\det \bm{J}|$ decrease.

\begin{figure*}[t]
\centering
\includegraphics[width=\textwidth]{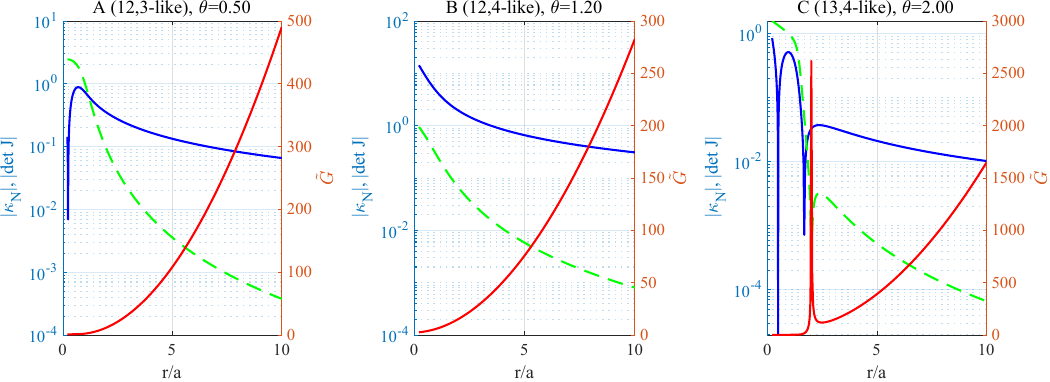}
\caption{Twin-axis radial sweeps of $|\kn|$ (left axis, solid blue), $|\det \bm{J}|$ (left axis, dashed green), and $\tilde{G}$ (right axis, solid red) versus $\tilde{r}$, at three configurations representative of the deployment subsystems. Panel C shows a sharp $\kn$ dip at $\tilde{r} \approx 2$ where $\tilde{G}$ remains moderate, an off-diagonal $\kn$-bad / GDoP-acceptable instance.}
\label{fig:radial_sweeps}
\end{figure*}

\subsection{Gate complementarity}

Table~\ref{tab:kappa_vs_gdop_concrete} reports $\kn$ and $\tilde{G}$ at one representative atlas target per quadrant of the $2 \times 2$ classification, at the matched fifth-percentile thresholds ($|\kn| < 5.47 \times 10^{-3}$, $\tilde{G} > 2.34 \times 10^{3}$). The four rows span markedly different parts of $(\beta, \gamma, \tilde{r}, \theta)$ space, demonstrating that the four regimes are not co-located.

\begin{table}[t]
\caption{Atlas-derived $\kn$ and $\tilde{G}$ at one representative target per quadrant of the $2 \times 2$ classification, at the atlas-wide matched $5$th-percentile thresholds $|\kn| < 5.47 \times 10^{-3}$, $\tilde{G} > 2.34 \times 10^{3}$. }
\label{tab:kappa_vs_gdop_concrete}
\centering
\setlength{\tabcolsep}{4pt}
\begin{tabular}{lcccccc}
\hline
Regime & $(\beta, \gamma)$ & $\tilde{r}$ & $\theta$~(deg) & $|\kn|$ & $|\det \bm{J}|$ & $\tilde{G}$ \\
\hline
both-good & $(-0.87, +0.55)$ & 0.86 & 262.4 & 6.45e-02 & 2.17e-03 & 1.52e+02 \\
$\kappa$-bad only & $(+1.50, -1.34)$ & 0.58 & 244.1 & 3.32e-03 & 1.44e-03 & 3.01e+02 \\
GDoP-bad only & $(-0.24, -0.08)$ & 8.76 & 311.2 & 2.08e-02 & 2.87e-05 & 2.78e+03 \\
both-bad & $(-0.87, +1.03)$ & 0.98 & 335.6 & 2.24e-03 & 1.18e-05 & 7.46e+03 \\
\hline
\end{tabular}
\end{table}

Treating $\kn$ and $\tilde{G}$ as thresholded gates, Table~\ref{tab:gate_roc} reports the ROC summary over the full atlas, where ``true positive'' means $|\det \bm{J}|$ falls in the bottom 5\% (truly Jacobian-singular). The disjoint $2 \times 2$ classification of all atlas points at matched 5\% thresholds yields $93.0\%$ ``both good'', $2.0\%$ ``$\kn$-bad only'' (caught by $\kappa$ alone), $2.4\%$ ``GDoP-bad only'' (caught by GDoP alone), and $2.6\%$ ``both bad''. Approximately $4.4\%$ of the atlas (the off-diagonal sum) is detected by exactly one of the two diagnostics, a quantifiable complementarity signal at scale.

\begin{table}[t]
\caption{Singularity gate ROC at matched 5\% percentile. By identity (T2.2), \(\Delta_N=0 \equiv \det J=0\), so a \(\Delta_N\)-gate is redundant with the GDoP gate and omitted.}
\label{tab:gate_roc}
\centering
\begin{tabular}{lcc}
\hline
Gate & TPR & FPR \\
\hline
$|\kappa_N| < 5.47e-03$ & 0.551 & 0.024 \\
$G > 2.34e+03$ & 0.900 & 0.015 \\
\hline
\end{tabular}
\end{table}

The $\tilde{G}$-gate has a direct operational reading through \eqref{eq:CEP50}: at $\sigma_t = 30$~ns, the threshold $\tilde{G} > 2.34 \times 10^{3}$ corresponds to $\mathrm{CEP}_{50} \gtrsim 16$~km. The $\kappa$-gate is independent of this operational reading: it flags target positions where the analytical solver's $K$-quadratic structure itself degenerates, irrespective of the noise-floor-set position-error budget.

\section{Case Study: An Operational Four-Node Array}
\label{sec:case-study}

We validate the framework on an operational four-node TDoA array whose three two-dimensional subsystems are formed by selecting any three of the four sensors with the CRS held fixed at the first sensor. The case study (i) instantiates the dimensionless atlas analysis on a deployed geometry with empirically estimated noise, and (ii) tests whether the theory-predicted threshold of Corollary~\ref{cor:threshold} agrees, at the deployment noise level, with a direct downstream Monte~Carlo measurement of localization error. The dimensionless geometry of the three subsystems is reported in Table~\ref{tab:geometry}; the three subsystems span markedly different geometries, with B exhibiting the most extreme aspect ratio ($|\beta| > 1$).

\begin{table}[t]
\caption{Per-subsystem dimensionless geometry $(\beta, \gamma)$ for the four-node operational array's three two-dimensional subsystems.}
\label{tab:geometry}
\centering
\begin{tabular}{lcc}
\hline
Subsystem & $\beta$ & $\gamma$ \\
\hline
A & -0.157 & 0.983 \\
B & -1.313 & -0.620 \\
C & -0.406 & 1.400 \\
\hline
\end{tabular}
\end{table}

\paragraph{Why three-sensor 2D subsystems rather than the four-sensor 3D solver} The deployment's four receivers sit at altitudes $510$ to $608$~m over a horizontal baseline of order $30$~km, an aspect ratio of roughly $1:300$. The four-node 3D analytical solver under this near-coplanar configuration is poorly conditioned in the vertical direction: the median dimensionless 3D GDoP across the operational range envelope is of order $290$, two orders of magnitude above the 2D values. Two-dimensional analysis on a three-sensor subsystem is therefore the operationally faithful granularity, with vertical recovered from the ADS-B-broadcast altitude or an independent altimeter - standard practice for wide-area multilateration deployments with co-sited altitudes \cite{eurocontrol2005wam}.

\subsection{Per-subsystem threshold prediction}

For each subsystem, we sample $1000$ random target positions in the operational range envelope $\tilde{r} \in [0.5, 10]$, compute $\sigma_\kappa$ from Theorem~\ref{thm:T3}, and apply Corollary~\ref{cor:threshold} with $k = 3$ to obtain the predicted threshold $\varepsilon_s^{\text{pred}} = 3 \cdot \operatorname{median}(\sigma_\kappa)$.

The per-sensor ToA noise $\sigma_t$ is estimated from $9$ ADS-B-truth-paired OTA captures of distinct Mode-S transponders, retained from a larger pool by an iid-consistency check: a capture is kept only if all three per-pair residual standard deviations fall below $500$~ns and agree within a factor of $1.5$. For each retained capture, the predicted sensor-pair TDoA is computed from the four receivers' published Earth-Centered, Earth-Fixed (ECEF) positions and the ADS-B-reported aircraft Lat/Long/Alt; raw integer-tick TOA differences are converted to seconds, and a per-pair static clock skew is removed as a constant offset (consistent with un-compensated cable-delay differences). The post-fit residual standard deviation is $\sigma_{\text{TDoA}, ij}$, and the per-sensor noise follows as $\sigma_t = \sigma_{\text{TDoA}, ij} / \sqrt{2}$.

Variance-pooling across $9$ captures ($3305$ samples per pair) and averaging across the three pairs gives
\begin{equation}
\sigma_t \;=\; 283~\text{ns} \quad (95\% \text{ CI}:  [278.3,\ 284.5]~\text{ns})
\label{eq:sigma_t_emp}
\end{equation}
with per-pair values $286$, $285$, $277$~ns (Figure~\ref{fig:sigma_t_clean_set}). The $95\%$ CI is computed by non-parametric cluster bootstrap over the $9$ aircraft ($10^4$ resamples). The per-pair values agree within $\sim$3\% and the per-capture cluster is tightly contained in the $270$--$290$~ns band across all $9$ aircraft, supporting the independent-per-sensor-noise model and indicating $\sigma_t = 283$~ns is a property of the deployment's timestamping infrastructure rather than any single transponder. This value is used in all subsequent runs. Applying the threshold procedure at $\sigma_t = 283$~ns gives $\varepsilon_A^{\text{pred}} = 4.38 \times 10^{-2}$, $\varepsilon_B^{\text{pred}} = 8.57 \times 10^{-2}$, $\varepsilon_C^{\text{pred}} = 4.06 \times 10^{-2}$. Subsystem B's predicted threshold is approximately double those of A and C, consistent with B's more extreme dimensionless geometry ($|\beta| > 1$) producing a larger $\sigma_\kappa$ noise floor per Theorem~\ref{thm:T3}.

\begin{figure}[t]
\centering
\includegraphics[width=\columnwidth]{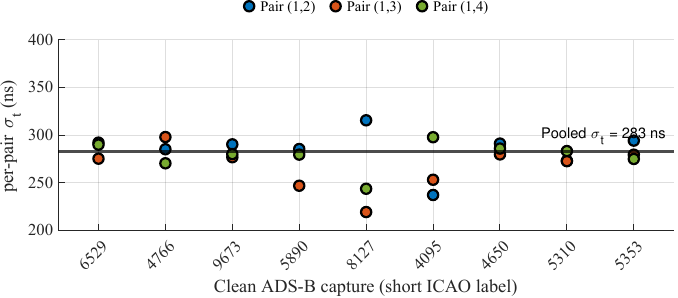}
\caption{Per-pair $\sigma_t$ for the $9$ ADS-B-paired captures. Horizontal line: pooled $\sigma_t = 283$~ns from \eqref{eq:sigma_t_emp}.}
\label{fig:sigma_t_clean_set}
\end{figure}

\begin{figure*}[t]
\centering
\includegraphics[width=\textwidth]{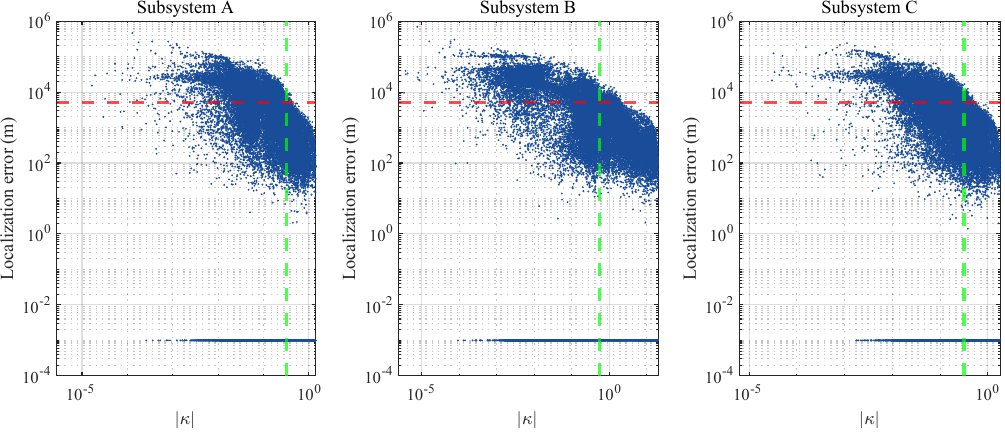}
\caption{Per-subsystem localization error vs $|\kn|$ from $10^4$-trajectory Monte-Carlo. Red dashed: $5$~km tolerance. Green dashed: empirical $\varepsilon^{\text{meas}}$ at $95\%$ reliability.}
\label{fig:error_vs_kappa}
\end{figure*}

\begin{figure}[t]
\centering
\includegraphics[width=\columnwidth]{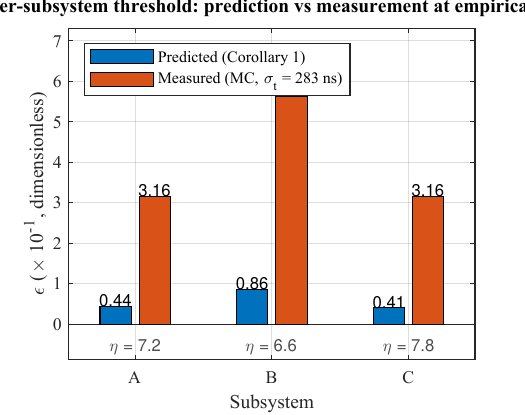}
\caption{Theory-predicted (Corollary~\ref{cor:threshold}) versus Monte-Carlo-measured $\varepsilon$ per subsystem at empirical $\sigma_t = 283$~ns. Theory and MC agree on per-subsystem ordering ($B > A \approx C$ in both); $\varepsilon^{\text{meas}}/\varepsilon^{\text{pred}} \approx 7$ across all three subsystems.}
\label{fig:eps_bar}
\end{figure}

\subsection{Operational threshold from Monte~Carlo localization-error analysis}

We now run a separate, downstream-coupled Monte~Carlo procedure that measures the \emph{operational} threshold: the $|\kn|$ value below which downstream localization error exceeds an operational tolerance. For each subsystem, twelve trajectories within $80$~km of the array centroid, with five independent noise realizations each, are perturbed at the empirical $\sigma_t = 283$~ns and passed through the analytical solver. Per-sample localization error is scattered against $|\kn|$ (Figure~\ref{fig:error_vs_kappa}). The measured threshold $\varepsilon^{\text{meas}}$ is the smallest $|\kn|$ cutoff such that the survivor population achieves the $5$~km tolerance with $\geq 95\%$ reliability.

The characteristic L-shape is present in all three panels: samples with large $|\kn|$ cluster at low error, while small-$|\kn|$ samples fan upward into the $10^4$ to $10^5$~m error tail, one to two orders of magnitude above the $5$~km tolerance (red dashed). The L-elbow positions, where the upper error envelope first crosses tolerance, locate the per-subsystem $\varepsilon^{\text{meas}}$ (green vertical).

Figure~\ref{fig:eps_bar} compares predicted versus measured thresholds. Subsystem~A: $\varepsilon^{\text{pred}} = 4.38 \times 10^{-2}$, $\varepsilon^{\text{meas}} = 3.16 \times 10^{-1}$ ($\eta_A \approx 7.2$); subsystem~B: $8.57 \times 10^{-2}$ vs $5.62 \times 10^{-1}$ ($\eta_B \approx 6.6$); subsystem~C: $4.06 \times 10^{-2}$ vs $3.16 \times 10^{-1}$ ($\eta_C \approx 7.8$). Two features stand out. First, theory and Monte~Carlo agree on per-subsystem ordering: both rank B as the most restrictive ($\varepsilon^{\text{pred}}_B > \varepsilon^{\text{pred}}_A \approx \varepsilon^{\text{pred}}_C$, and likewise for $\varepsilon^{\text{meas}}$). The algebraic gate of Corollary~\ref{cor:threshold} therefore predicts the correct operational ordering at deployment noise. Second, defining
\begin{equation}
\eta_s \;\equiv\; \frac{\varepsilon^{\text{meas}}_s}{\varepsilon^{\text{pred}}_s}, \qquad
\eta \;\equiv\; \frac{1}{|\mathcal{S}|} \sum_{s \in \mathcal{S}} \eta_s,
\label{eq:eta_calibration}
\end{equation}
the per-subsystem $\eta$ values lie in a narrow band of width $\sim$$17\%$, with $\eta \approx 7.2$ at the deployment $\sigma_t$.

$\eta$ is the case-study counterpart, at the deployment noise level, to $\mathrm{CEP}_{50} \approx 0.75 \, c \, \sigma_t \, \tilde{G}$ of \eqref{eq:CEP50}, which converts dimensionless GDoP into operational $\mathrm{CEP}_{50}$. The two relations together let an operator translate $\kn$ and GDoP gates between their algebraic and operational forms. The two thresholds remain valid for distinct operational decisions: $\varepsilon^{\text{pred}}$ flags samples on which the analytical solver's algebraic conditioning is in question regardless of the downstream error budget; $\varepsilon^{\text{meas}}$ flags samples whose downstream localization error would exceed tolerance. Joint deployment in the $2 \times 2$ classification of \S\ref{sec:bridge-to-gdop} gives finer per-sample reporting than either alone.

\subsection{Sensitivity to $\sigma_t$ and tolerance choice}
\label{sec:sensitivity}

Theorem~\ref{thm:T3} predicts that $\sigma_\kappa$ scales linearly with $\sigma_t$, so $\varepsilon^{\text{pred}}$ also scales linearly; Table~\ref{tab:sensitivity} reports predicted thresholds at three representative noise levels covering typical ADS-B and SDR-receiver timestamping accuracies. The per-subsystem ordering is invariant. The same $\sigma_t$ axis maps to $\mathrm{CEP}_{50}$ through \eqref{eq:CEP50}, giving per-subsystem $\mathrm{CEP}_{50}$ at $\sigma_t = 30$~ns of order a few kilometres for A and C and a few tens of kilometres for B.

\begin{table}[t]
\caption{Sensitivity of predicted threshold $\varepsilon^{\text{pred}}$ to per-sensor ToA noise $\sigma_t$. Threshold scales linearly with $\sigma_t$ per Theorem~\ref{thm:T3}; per-subsystem ordering is invariant. Bottom row: deployment-specific empirical $\sigma_t$ from \eqref{eq:sigma_t_emp}.}
\label{tab:sensitivity}
\centering
\begin{tabular}{lccc}
\hline
$\sigma_t$ & Subsys A & Subsys B & Subsys C \\
\hline
10 ns        & $1.55 \times 10^{-3}$ & $3.03 \times 10^{-3}$ & $1.43 \times 10^{-3}$ \\
30 ns        & $4.64 \times 10^{-3}$ & $9.09 \times 10^{-3}$ & $4.30 \times 10^{-3}$ \\
100 ns       & $1.55 \times 10^{-2}$ & $3.03 \times 10^{-2}$ & $1.43 \times 10^{-2}$ \\
283 ns (emp.)& $4.38 \times 10^{-2}$ & $8.57 \times 10^{-2}$ & $4.06 \times 10^{-2}$ \\
\hline
\end{tabular}
\end{table}

The Monte-Carlo-measured threshold $\varepsilon^{\text{meas}}$ depends on the downstream tolerance through the L-shape of Fig.~\ref{fig:error_vs_kappa}, with sub-linear dependence because the vertical edge spans many decades. The calibration ratio $\eta$ is, however, deployment-specific in a stricter sense than $\varepsilon^{\text{pred}}$: while $\sigma_\kappa$ scales linearly with $\sigma_t$, $\varepsilon^{\text{meas}}$ does not, because the L-elbow's position reflects a balance between $\sigma_\kappa$-coupled and GDoP-coupled error contributions whose mix shifts with $\sigma_t$. As a numerical illustration, re-running the MC at $\sigma_t = 30$~ns yields $\varepsilon^{\text{meas}}$ values $\{1.8 \times 10^{-2}, 5.6 \times 10^{-3}, 1.8 \times 10^{-2}\}$ for A, B, C; the per-subsystem ordering changes to $A \approx C > B$ and the meas-to-pred ratios fan out to $\{3.9, 0.6, 4.2\}$. The deployment-meaningful calibration is therefore the one run at the deployment $\sigma_t$ directly.

\section{Discussion}
\label{sec:discussion}

\subsection{Practical implications for receiver-array design}
Identity \eqref{eq:T21} expresses $\det \bm{J}$ as a product of $\det \bPbar$ (a geometric volume of the canonical sensor matrix) and $(\kn K^2 - \|\bA\|^2)$ (a target-dependent term). Receiver-array design therefore admits a clean decomposition: $\det \bPbar$ is set by sensor placement (small $\det \bPbar$ means a near-degenerate baseline configuration to be avoided), while the target-dependent factor sets which regions of operational airspace are Jacobian-good or Jacobian-bad. Similarly, identity \eqref{eq:T22} says that any branch-merge target is automatically a Jacobian-singular target, so an analytical solver that picks a single root via geometric disambiguation can use the GDoP layer alone to flag ambiguous samples, avoiding the need for an explicit discriminant gate. In operational practice, the GDoP-layer threshold can be expressed equivalently in $\mathrm{CEP}_{50}$ at the deployment's nominal $\sigma_t$ via \eqref{eq:CEP50}; the algebraic-complementarity claim between $\kn$ and $\tilde{G}$ transfers directly to $\kn$ and $\mathrm{CEP}_{50}$ without modification. The $\kappa$-gate adds a category of pathology - the branch-divergence regime of \S\ref{sec:singularity-structure} - that no purely position-error-budget metric ($\tilde{G}$, $\mathrm{CEP}_{50}$, major-axis $\sigma$) can reach.

\paragraph{Relation to algebraic-geometric prior work} The two singularity loci of \S\ref{sec:singularity-structure} are not new objects: the leading-coefficient locus ($\kn = 0$) and the related Jacobian-discriminant locus were characterized for the planar three-receiver case by Compagnoni and Notari \cite{compagnoni2014inverse, compagnoni2014arxiv}, with related range-map and three-receiver three-dimensional results in \cite{compagnoni2017jns} and a statistical treatment of the singular set in \cite{compagnoni2016arxiv}. Our contribution is twofold: (i) closed-form identities (T2.1, T2.2) factoring $\det \bm{J}'$ and $\dn$ in solver-state coordinates $(\kn, K, \|\bA\|^2)$ for general dimension $N$, directly evaluable from TDoA data without prior knowledge of the source position; and (ii) the operationalization of $\kn$ as a runtime geometry-quality metric reported alongside GDoP, with a closed-form noise sensitivity $\sigma_\kappa$ and a complementarity statistic measured across a parameterized atlas.

\subsection{Limitations}
The current treatment is for time-synchronized receiver arrays in line-of-sight conditions. Extensions to unsynchronized arrays \cite{sonnleitner2025, linwang2024}, to non-line-of-sight conditions \cite{dickerson2025}, and to over-the-horizon multistatic-radar configurations \cite{nassosanti2024} are natural future directions. The calibration constant $\eta \approx 7$ between the algebraic and operational thresholds is deployment-specific; generalizing $\eta$ to other deployments requires a comparable per-deployment MC, since the L-shape elbow's $\sigma_t$-sensitivity is non-linear in the singular regime even though $\sigma_\kappa$ itself scales linearly with $\sigma_t$.

\section{Conclusion}
\label{sec:conclusion}

We introduced $\kn = \|\bB\|^2 - 1$, the leading coefficient of the analytical-TDoA $K$-quadratic, as a geometry-quality diagnostic for multilateration. Two algebraic identities ((T2.1) factoring $\det \bm{J}'$ and (T2.2) the discriminant $\dn = (\kn K^2 - \|\bA\|^2)^2 / (4K^2)$) tie $\kn$ to the Jacobian determinant and the $K$-quadratic discriminant, establishing exactly two distinct codimension-one singularity loci: branch divergence ($\kn = 0$) and Jacobian / branch-merge ($\det \bm{J} = 0 \equiv \dn = 0$), the latter captured by GDoP. A CRLB-linked closed-form noise sensitivity $\sigma_\kappa$ under the standard Gaussian ToA model, validated against Monte~Carlo to $2\%$ median relative error, supports the three-sigma threshold rule of Corollary~\ref{cor:threshold}.

An empirical atlas over $1.44 \times 10^6$ dimensionless samples confirms both identities at machine precision and establishes that $\kn$ and GDoP carry substantially independent geometry-quality information: per-geometry Pearson correlation between $\log_{10}|\kn|$ and $\log_{10}\tilde{G}$ ranges from $-0.88$ to $-0.56$ (median $-0.71$), with approximately $4.4\%$ of the atlas falling in off-diagonal regimes detected by exactly one of the two diagnostics. The case study on a four-node operational array validates the framework against direct downstream Monte~Carlo at the deployment's empirically estimated $\sigma_t = 283$~ns: theory and Monte~Carlo agree on per-subsystem ordering (B most restrictive in both), with a near-constant calibration ratio $\eta \approx 7$ anchoring a deployment-specific bridge between the algebraic and operational threshold definitions, analogous in spirit to the standard $\mathrm{CEP}_{50}$--GDoP relation. The framework provides receiver-array designers with two distinct geometry-quality layers and a $2 \times 2$ classification of target positions that refines the GDoP-only picture, applicable to any closed-form analytical TDoA system. Extensions to wider-aperture arrays, unsynchronized configurations, non-line-of-sight environments, and FDoA-augmented hybrid measurements are immediate next directions.


\begin{IEEEbiography}[{\includegraphics[width=1in,height=1.25in,clip,keepaspectratio]{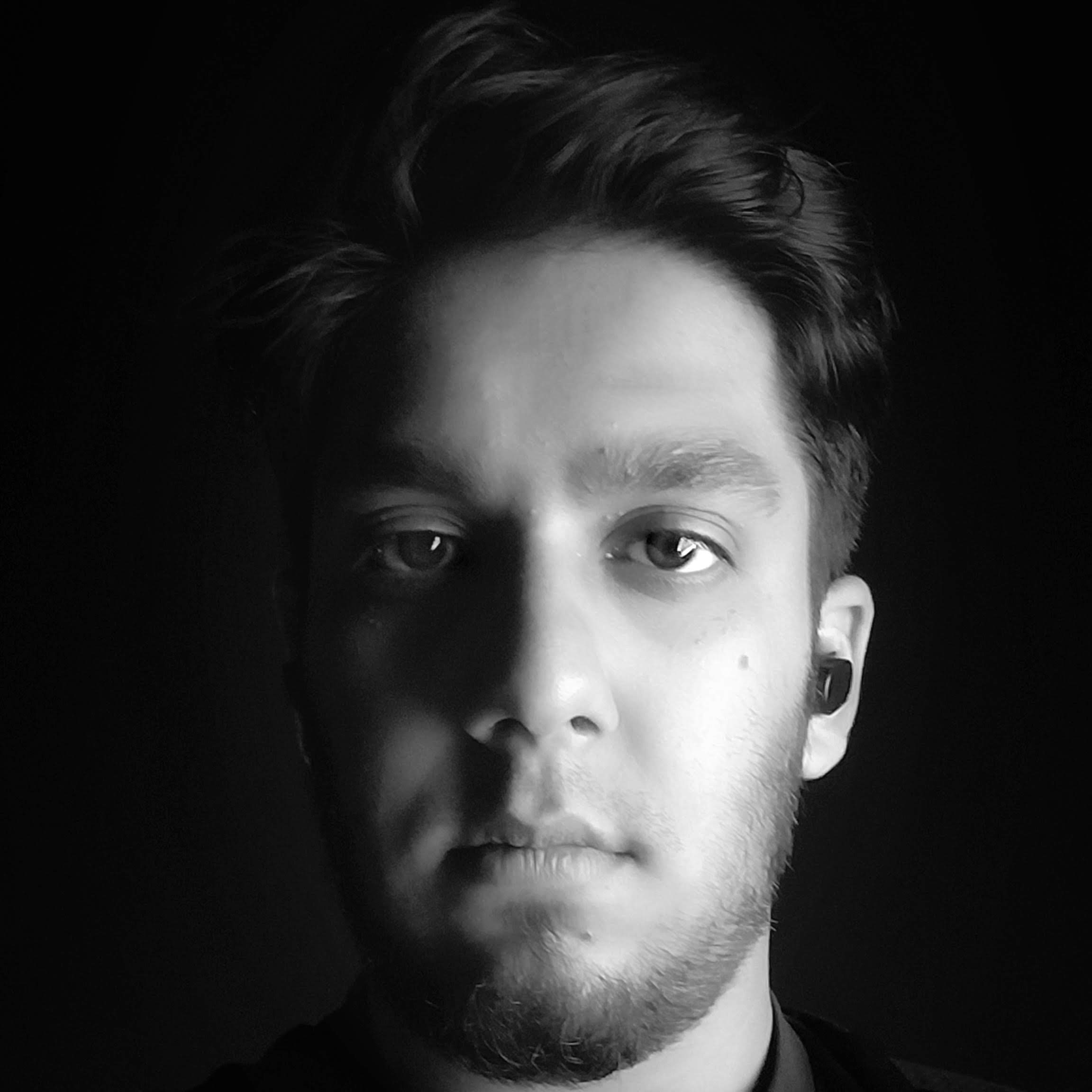}}]{Abeer Nasir Chaudhry}
received the B.E. degree in Avionics Engineering from the National University of Sciences and Technology (NUST), Islamabad, Pakistan, in 2019, and the M.S. degree in Electrical Engineering from the American University of Sharjah, Sharjah, U.A.E., in 2021. Since 2021, he has been working as a DSP Simulation Engineer with the National Aerospace Science and Technology Park (NASTP), Islamabad, Pakistan. His research interests include statistical signal processing, remote sensing, and automotive radar signal processing. Mr.~Chaudhry was the recipient of the Rector's Gold Medal for the Best Undergraduate Final Year Project in 2019. He is also a recipient of the Marie Sk\l{}odowska-Curie Actions (MSCA) Doctoral Networks (DN) Fellowship and the University International Postgraduate Award (UIPA), University of New South Wales (UNSW), Australia.
\end{IEEEbiography}

\begin{IEEEbiography}
[{\includegraphics[width=1in,height=1.25in,clip,keepaspectratio]{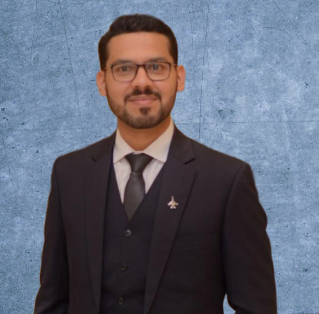}}]
{Salman Liaquat} (Member IEEE) completed his B.E. Avionics from National University of Sciences and Technology (NUST), Islamabad, Pakistan, in 2010, Master of Science in Avionics Engineering from Air University, Islamabad, Pakistan, in 2019, and Ph.D. in Electrical and Electronic Engineering from Universiti Sains Malaysia, Penang, Malaysia in 2026. He is currently working with the National Aerospace Science and Technology Park (NASTP), Islamabad, Pakistan. His research interests include digital signal processing and its applications in radars. With over ten years of professional experience in avionics and radar systems, he has authored several publications in international journals and presented at peer-reviewed conferences. He is the recipient of IEEE AESS Engineering Scholarship for Graduate Students in 2025, and IEEE Signal Processing Society Scholarship in 2023 and 2024. He is professionally registered as a Professional Engineer with Pakistan Engineering Council, Chartered Engineer (CEng) with the Engineering Council UK, and is a Member of the Royal Aeronautical Society (MRAeS), United Kingdom.
\end{IEEEbiography}

\begin{IEEEbiography}
[{\includegraphics[width=1in,height=1.0in,clip,keepaspectratio]{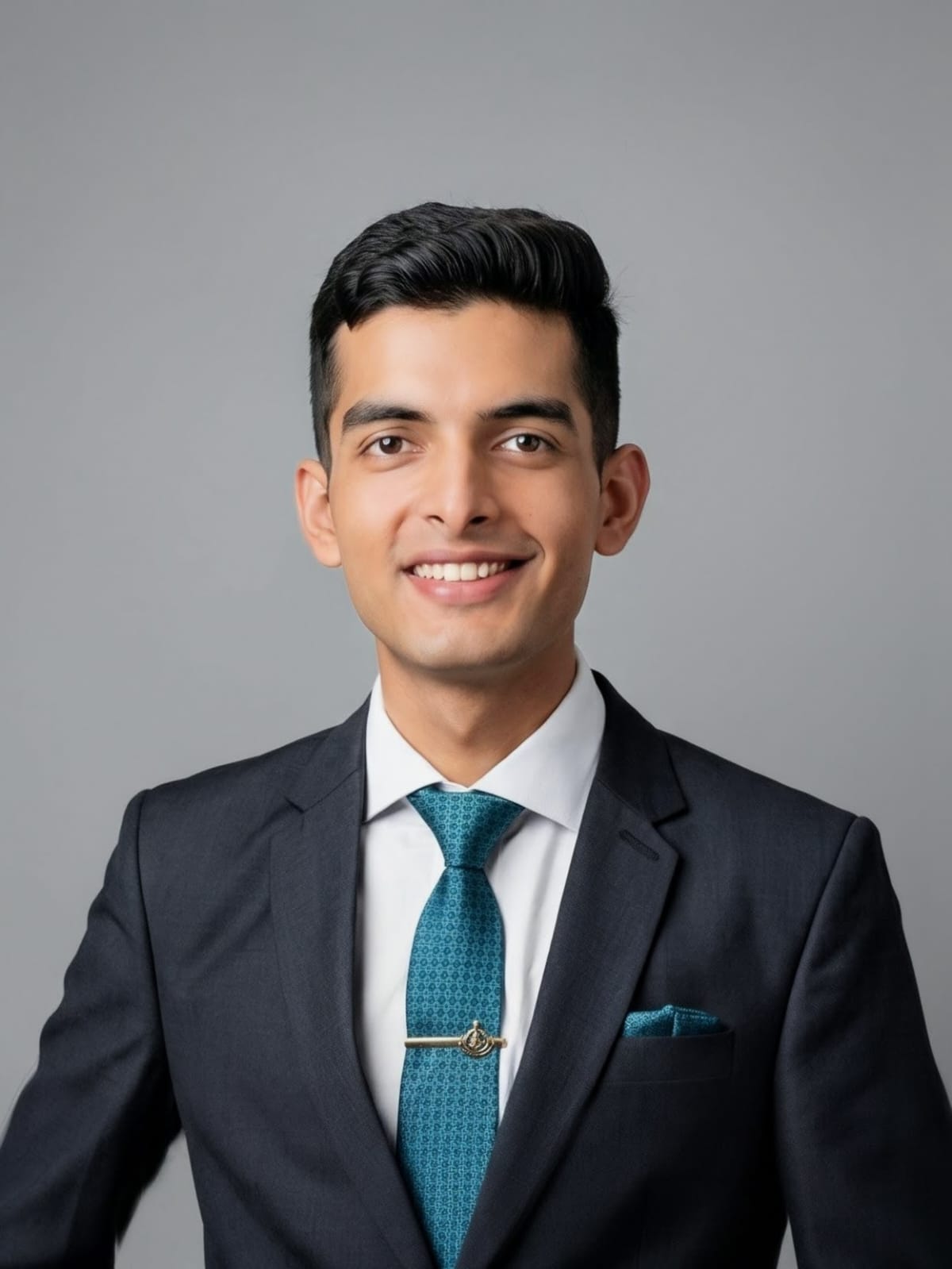}}] 
{Muhammad Mohsin Khadim} completed his B.E. in Avionics from the National University of Sciences and Technology (NUST), Pakistan, in 2024. He has over two years of professional experience in avionics and digital signal processing. His research interests lie in fine-tuning large language models, natural language processing, and applying classical signal-processing methods to deep neural networks.
\end{IEEEbiography}

\end{document}